\documentclass[12pt]{article}

\usepackage{algorithm}
\usepackage{algorithmic}
\usepackage{amsmath}
\usepackage{amssymb}
\usepackage{amsthm}

\newtheorem{theorem}{Theorem}
\newtheorem{definition}[theorem]{Definition}
\newtheorem{corollary}[theorem]{Corollary}
\newtheorem{proposition}[theorem]{Proposition}

\title{The Discrete Sell or Hold Problem with Constraints on Asset Values}
\author{\small Ye Du\\
       {\small Department of Financial Engineering, Southwestern University of Finance and Economics}\\
       {\small duye@swufe.edu.cn}}

\date{}

\begin{document}
\maketitle \abstract{The discrete sell or hold problem (DSHP), which
is introduced in \cite{H12}, is studied under the constraint that
each asset can only take a constant number of different values. We
show that if each asset can take only two values, the problem
becomes polynomial-time solvable. However, even if each asset can
take three different values, DSHP is still NP-hard.  An
approximation algorithm is also given under this setting.

\section{Introduction}
There are three key factors in asset management, which are {\it
asset allocation}, {\it security selection} and {\it timing}
\cite{W11}. A common scenario faced by an asset management team is
to meet various kinds of capital requirements, such as customer
withdraw needs or regulation requirements, at the end of a fiscal
year. In order to achieve those capital goals, an asset manager may
need to sell part of its portfolio holdings to generate as much
capital as possible. Dealing with this scenario, an asset manager
should take {\it security selection} (which assets to sell), and
{\it timing} (when to sell an asset) into consideration.

Since the future price of a financial asset is a stochastic process,
the {\it timing} issue poses a lot of challenges to asset managers.
Instead of studying when to sell an asset in a continuous time
setting, a simplified model is to study whether we should sell an
asset now or at a specified future date. The problem, which is
called {\it the discrete sell or hold problem (DSHP)}, is introduced
by He {\it et. al.} \cite{H12}. It can be modeled as a two-stage
stochastic combinatorial optimization problem. It is shown that DSHP
is NP-hard to solve.

However, the model in \cite{H12} does not impose any constraints on
the possible values that an asset can take. A natural conjecture is
that the number of different values an asset can take would have
great impact on the complexity of DSHP. Following this idea, we can
further simplify the model by assuming that each asset can take only
a constant number of different values. In particular, we study the
case when an asset can take only three different values. This case
coincides with the idea of {\it binomial tree model}, which is a
standard finance textbook model \cite{FE00} for asset pricing. In
the binomial tree model, the current value of asset is $v$; in the
future stage, it can take two possible values: a high value scenario
$v_u$ and a low value scenario $v_d$. Thus, it is very interesting
to study the discrete sell or hold problem under this simplified
model. Furthermore, the case that an asset can take only two
possible values is also studied.

The organization of the paper is as follows: the mathematical
formulation of the discrete sell or hold problem is given in section
\ref{sec:f}; the complexity of the problem with constraints on asset
values is studied in section \ref{sec:c}; an approximation algorithm
is presented in section \ref{sec:aa}; the conclusion is in section
\ref{sec:final}.

 %{\ignore A specific example is
%that: a fund currently holds a basket of $n$ assets. In order to
%generate more capital, the fund manager is to sell $k$ assets by the
%end of next year.}

\section{The Problem Formulation}{\label{sec:f}}
The discrete sell or hold problem (DSHP) deals with a decision
problem of two stages. There are $n$ assets $A=\{1,2,...,n\}$ to
manage. At the first stage, the value of each asset $c_i$ is known.
However, at the second stage, there are $m$ scenarios. Under each
scenario $j$, every asset $i$ has value $f_{ij}$. We plan to sell
$k$ assets to generate as much revenue as possible. Each asset can
be hold or sold. If it is sold, we have to decide whether we should
sell it at the first stage or the second stage. The problem can be
formulated as an integer programming problem as follows.
\[
\begin{array}{lcr}
&\max\sum\limits_{i=1}^nc_ix_i+\sum\limits_{j=1}^mp_j\sum\limits_{i=1}^nf_{ij}y_{ij}&\\
& s.t. \sum\limits_{i=1}^nx_i+\sum\limits_{i=1}^ny_{ij}=k&j=1,...,m\\
&x_i+y_{ij}\leq 1& i=1,...,n,j=1,...,m\\
&x_i \in \{0,1\},y_{ij} \geq 0&
\end{array}
\]

Since the discrete sell or hold problem deals with uncertainty in
the future state, it belongs to a category of problems called {\it
combinatorial stochastic programming problem}, which is extensively
studied in computer science and operation research literatures
\cite{DRS05,GRS04, GRS07, KS06, RS06}.

\section{DSHP with Constraints on Asset Values} \label{sec:c}
\subsection{The Case when Assets can take only two values}
\begin{proposition}If $c_i \in \{v_{max},v_{min}\}, i=1,...,n$ and $f_{ij}\in
\{v_{max},v_{min}\}, i=1,...,n,j=1,...,m$, DSHP is solvable in
$O(nm)$ time.\end{proposition}
\begin{proof} Suppose $C=\{i|c_i=v_{max}\}$. We will construct an optimal solution to the DSHP as follows:
if $|C|<k$, all the assets in $C$ are sold at the first stage; for
each scenario at the second stage, the most valuable $k-|C|$ assets
are sold; Otherwise, if $|C| \geq k$, arbitrary $k$ assets in $C$
are sold at the first stage while nothing is sold at the second
stage.

We will show that the above solution is an optimal one. For $i \in
\{1,..,n\}$, if $c_i=v_{min}$, the revenue of selling asset i at the
second stage is always at least as good as selling it at the first
stage. Thus, in order to maximize the total revenue, we don't have
to sell asset i at the first stage. If $c_i=v_{max}$, suppose that
in an optimal solution $\mathcal{O}$, asset i is not sold at the
first stage, and some assets are sold for each scenario $j \in
\{1,...,m\}$ at the second stage. Then, we can construct another
solution $\mathcal{O}^{'}$ that asset i is sold at the first stage
while one asset $r_j$ is removed from the sold asset list of the
optimal solution $\mathcal{O}$ for each scenario $j$. For any choice
of $r_j$, the solution $\mathcal{O}^{'}$ is at least as good as
$\mathcal{O}$. Thus, if $c_i=v_{max}$, asset i should be sold at the
first stage to achieve the optimal revenue. Combining both cases, it
is easy to see that the solution we construct above is an optimal
solution.
\end{proof}

\begin{corollary}For an asset i that $c_i < \sum\limits_{j=1}^m p_jf_{ij}$, it should not be sold at the first stage in an optimal solution of
DSHP.\end{corollary}

\subsection{The Case when Assets can take three values}
\begin{definition}
Given a graph $G=(V,E)$, the dominating set of $G$ is a set of
vertices $D$ such that $\forall v \in V$, either $v \in D$ or
$\exists u \in D$ such that $(v,u) \in E$.
\end{definition}
The problem of finding the minimum dominating set for a regular
planar graph with degree 4, which is called the {\it MDS-RPG4}
problem, is NP-hard \cite{GJ79}. We will reduce the MDS-RPG4 problem
to DSHP and show that solving DSHP is NP-hard as well.

The construction of the reduction is: There are $n$ assets and $n$
scenarios, each of which has probability $\frac{1}{n}$ to realize.
For each asset $i$, its value at the first stage is $1$, while its
value at the second stage under each scenario is defined according
to the graph $G=(V,E)$ as following:
\begin{itemize}
\item $\forall i, f_{ii}=1-B$;
\item $\forall i,j$, if $(i,j)\in E$, $f_{ij}=1-B$;
\item $\forall i,j$, if $(i,j)\notin E$, $f_{ij}=1+S$.
\end{itemize}
Here we make $B$, $M$ and $S$ satisfies the following condition:
\begin{itemize}
%\item $B=M$;
\item $\frac{d}{n-d} < \frac{S}{B} < \frac{d+1}{n-d-1}$.
\end{itemize}
Where $d$ is the degree of the regular graph $G$. It is easy to
check that $\frac{d}{n-d} < \frac{d+1}{n-d-1}$ holds for any $d \in
[0,n-1)$. In particular, $d=4$ here. Thus, we have a matrix $M$ like
below:
\[
\begin{array}{|c|c|cccc}
&&1&2&...&n\\
\hline
1&1&.&.&...&.\\
2&1&.&.&...&.\\
...&1&...&...&...&...\\
n&1&.&.&...&.\\
\end{array}
\]
Note that in the instance of DSHP, each asset can take only three
values $\{1,1-B,1+S\}$.
\begin{theorem}DSHP is NP-hard, even if all the
assets can take only three different values.\end{theorem}
\begin{proof}We will reduce the MDS-RPG4 problem to
DSHP with $k=n-1$,where $n=|V|$. W.L.O.G, we assume that graph
$G$ is connected.% Moreover, let $d$ be the degree of the regular
%graph $G$. Note that here $d=4$.

%Suppose $D$ is a minimum dominating set of $G$, we will show that
%selling asset $i \in D$ at the first stage will induce an optimal
%solution for DSHP.

Suppose in an optimal solution the set of $V\setminus D$ assets are
sold at the first stage, we will show that $D$ is a minimum
dominating set. First, we show that $D$ must be a dominating set.
Otherwise, there is a node $i$ in $V\setminus D$ that is not
dominated by any node in $D$. Since the graph is connected, $i$
should be connected to some node $h$ in $V\setminus D$. We will
construct another solution for DSHP: sell assets in $V\setminus (D
\cup \{h\})$ at the first stage. Since asset $h$ is not sold at the
first stage, we define a {\it baseline perturbed solution} as
selling asset $h$ under each scenario at the second stage. Note that
in the $h$th column of the matrix (which represents the $h$th
scenario of the second stage), in the original solution, some asset
$g$ in $D$ is not sold under this scenario (remember $k=n-1$ and
there must be one asset left unsold). Thus, we can make further
improvement to the baseline perturbed solution by selling asset $g$
instead of asset $h$ under this scenario. Note that since $h$ is not
dominated by any node in $D$, $M_{gh}=1+S$ while $M_{hh}=1-B$. This
scenario is showed briefly in the following matrix.
\[
\begin{array}{|c|c|ccc}
&&...&h&...\\
\hline
...&1&.&...&...\\
h&1&...&1-B&...\\
...&1&...&...&...\\
g&1&...&1+S&...\\
...&1&...&...&...\\
\end{array}
\]
Let $R$ be the revenue generated by the original optimal solution
and $R'$ be the revenue generate by the new solution. We get,
\begin{eqnarray*}
R' & = & R-1+\frac{1}{n}[\sum\limits_{j}f_{ij}-(1-B)+(1+S)]\\
&=& R-1+\frac{1}{n}[n-dB+(n-d-1)S-B-(1-B)+(1+S)]\\
&=& R + \frac{1}{n}[-dB+(n-d)S] \\
&>& R
\end{eqnarray*}
This contradicts to the assumption that $D$ is not a dominating set
and selling assets in $V\setminus D$ induces an optimal solution.
Thus, $D$ must be a dominating set.

Next, we will show that if selling assets in $V\setminus D$ at the
first stage induces an optimal solution for DSHP, $D$ must be a
minimum dominating set. Otherwise, suppose there is another
dominating set $D'$ with $|D'|<|D|$. Since both $D$ and $D'$ are
dominating sets, at the second stage, for each scenario $s$, there
exists an asset(node) $g$ in $D$(or $D'$) such that $M_{gs}$ is
$1-B$. Let the revenues generated by both sets be $R$ and $R'$.
Thus, We get
\begin{eqnarray*}
R=n-|D|+\frac{1}{n}[|D|(n-dB+(n-d-1)S-B)-n(1-B)]
\end{eqnarray*}
and
\begin{eqnarray*}
R'= n-|D'|+\frac{1}{n}[|D'|(n-dB+(n-d-1)S-B)-n(1-B)]
\end{eqnarray*}
Thus,
\begin{eqnarray*}
R'-R&=&|D|-|D'|+\frac{1}{n}[(|D'|-|D|)(n-dB+(n-d-1)S-B)]\\
&=&|D|-|D'|+(|D'|-|D|)+\frac{1}{n}[(|D'|-|D|)(-(d+1)B+(n-d-1)S)]\\
&=&\frac{1}{n}[(|D'|-|D|)(-(d+1)B+(n-d-1)S)]\\
&>&0
\end{eqnarray*}
This contradicts to the assumption that $D$ is a minimum dominating
set and selling assets in $V\setminus D$ induces an optimal
solution. Thus, $D$ must be a minimum dominating set. Combining the
above two steps, we can get that in an optimal solution of DSHP, if
the set of $V\setminus D$ assets are sold at the first stage, $D$
must be a minimum dominating set.

Now we show the reverse direction. Suppose $D$ is a minimum
dominating set of $G$, we will show that selling assets in
$V\setminus D$ at the first stage will induce an optimal solution
$OPT$ for DSHP. Let the revenue generated by $OPT$ is $R$.
Otherwise, suppose there is another optimal solution $OPT'$ with
revenue $R'> R$. According to the argument above, $OPT'$ will induce
a smaller dominating set than $D$ (Note that dominating sets with
the same size will generate the same revenue). This contradicts to
the assumption that $D$ is a minimum dominating set. Therefore, if
$D$ is a minimum dominating set of $G$, selling assets in
$V\setminus D$ at the first stage will induce an optimal solution
for DSHP.

In all, we have reduced the MDS-RPG4 problem to DSHP. Thus, DSHP is
NP-hard, even if all the assets can take only three different
values.
\end{proof}

\section{The Approximation Algorithm} \label{sec:aa}
He {\it et. al.} \cite{H12} gives a
$\max\{\frac{1}{2},\frac{k}{n}\}$ algorithm for DSHP.  Here we
assume that each asset can take only three different values
$\{V_S,V_M,V_L\}$, which satisfies the following relationship $V_S <
V_M < V_L$. We present another approximation algorithm below.
\begin{algorithm}
\caption{A Heuristic for DSHP when assets can take only three
different values}
\begin{algorithmic}[1]
\STATE Suppose at the first stage, there are $t_L$ assets with value
$V_L$ and $t_M$ assets with value $V_M$ . If $t_L +t_M \leq k$, sell
all of them; otherwise, sell $k$ of them. \STATE If $(t_L+t_M)<k$,
under each scenario for the second stage, sell the top $k-(t_L+t_M)$
most expensive assets that have not been sold at the first stage.
\end{algorithmic}
\end{algorithm}
%Let $e_i=\sum\limits_{j}f_{ij}p_j$, which is the expected value of
%selling asset $i$ at the second stage.
%\begin{algorithm}
%\caption{A Heuristic for the DSHP when assets can only take three
%different values}
%\begin{algorithmic}[1]
%\STATE Suppose at the first stage, there are $t_L$ assets with value
%$V_L$. If $t_L \leq k$, sell all of them; otherwise, sell $k$ of
%them. \STATE If $t_L <k$, for any asset $i$ with $c_i > e_i$, sell
%it at the first stage until $k$ assets have been sold. Let $t_M$ be
%the number of assets sold at this step. \STATE If $(t_L+t_M)<k$,
%under each scenario for the second stage, sell the top $k-(t_L+t_M)$
%most expensive assets that have not been sold at the first stage.
%\end{algorithmic}
%\end{algorithm}
\begin{proposition}
Algorithm 1 is a $\frac{V_M}{V_L}$-approximation algorithm for DSHP
when assets can take only three different values $\{V_S,V_M,V_L\}$.
Here $V_S < V_M < V_L$.
\end{proposition}
\begin{proof}
Let $ALG_1(I)$ be the optimal value that Algorithm 1 gets on an
instance $I$ while $OPT(I)$ is the optimal value of the DSHP on
instance $I$. It is easy to see that every asset sold at the first
stage in Algorithm 1 has value of either $V_L$ or $V_M$.

If $(t_L+t_M)\geq k$, $ALG_1(I) \geq kV_M$ and $OPT(I) \leq kV_L$.
Therefore, $OPT(I) \leq \frac{V_L}{V_M} ALG_1(I)$. If $(t_L+t_M)<
k$, let $F_{OPT(I)}$ and $F_{ALG_1(I)}$(the size of it is $t_L+t_M$)
be the sets of assets that are sold in the optimal solution and in
Algorithm 1 respectively, while $V_{F_{OPT(I)}}$ and
$V_{F_{ALG_1(I)}}$ be the values generated in the first stage
respectively. It is easy to see that $F_{OPT(I)} \subseteq
F_{ALG_1(I)}$ \footnote{If some asset $i$ with value $V_S$ is sold
at the first stage in an optimal solution, we can make asset $i$
sold at the second stage instead of the first stage without changing
the optimal value.}. At the second stage, under each scenario, the
top $k-|F_{OPT(I)}|$ most expensive assets in
$S_{OPT(I)}=A/F_{OPT(I)}$ should be sold in the optimal solution.
Let $V_{S_{OPT(I)}}$ be the value generated at the second stage in
optimal solution. Moreover, the top $k-(t_L+t_M)$ most expensive
assets in $S_{ALG_1(I)}=A/F_{ALG_1(I)}$ should be sold in Algorithm
1. Note $S_{ALG_1(I)} \subseteq S_{OPT(I)}$. Let $V_{S_{ALG_1(I)}}$
be the value generated at the second stage by Algorithm 1. Moreover,
we define $d=(k-|F_{OPT(I)}|)-(k-|F_{ALG_1(I)}|)$. Then we can get
the following relationship:
$$
V_{S_{OPT(I)}} \leq d*V_L+V_{S_{ALG_1(I)}}
$$
Note that if the above relationship holds for each scenario of the
second stage, the relationship will hold for case when we consider
the expected value of revenue generated across all scenarios at the
second stage. Thus, we abuse the notation a little bit such that
$V_{S_{OPT(I)}}$ and $V_{S_{ALG_1(I)}}$ refer to both revenues for
each scenario and the expected revenue across all scenarios.
Therefore,
\begin{eqnarray*}
OPT(I)&=&V_{F_{OPT(I)}}+V_{S_{OPT(I)}} \\
&\leq&
\frac{V_L}{V_M}V_{F_{ALG_1(I)}}*\frac{|F_{OPT(I)}|}{|F_{ALG_1(I)}|}+d*V_L+V_{S_{ALG_1(I)}}\\
&\leq&
\frac{V_L}{V_M}[V_{F_{ALG_1(I)}}*\frac{|F_{OPT(I)}|}{|F_{ALG_1(I)}|}+d*V_M+V_{S_{ALG_1(I)}}]\\
&\leq&
\frac{V_L}{V_M}[\frac{V_{F_{ALG_1(I)}}}{|F_{ALG_1(I)}|}({|F_{OPT(I)}|}+d)+V_{S_{ALG_1(I)}}]\\
&=&\frac{V_L}{V_M}[V_{F_{ALG_1(I)}}+V_{S_{ALG_1(I)}}]\\
 &=& \frac{V_L}{V_M} ALG_1(I)
\end{eqnarray*}
Therefore, Algorithm 1 is a $\frac{V_M}{V_L}$-approximation
algorithm.
\end{proof}

In the next, we will show the approximation ratio $\frac{V_M}{V_L}$
is also tight. Consider the following example. There are four assets
to sell and three scenarios at the second stage. Let $k=1$.
Algorithm 1 will sell the second asset at the first stage, which
generates revenue $V_M$. However, in an optimal solution, we should
not sell any asset at the first stage. Nevertheless, we should sell
the most expensive one under each scenario at the second stage,
which generates revenue $V_L$.
\[
\begin{array}{|c|c|ccc}
&0&1&2&3\\
\hline
1&V_S&V_L&V_S&V_S\\
2&V_M&V_S&V_S&V_S\\
3&V_S&V_S&V_L&V_S\\
4&V_S&V_S&V_S&V_L\\
\end{array}
\]
Thus, the approximation ratio $\frac{V_M}{V_L}$ is tight.

Note that the $\frac{V_M}{V_L}$ approximation ratio could be better
than the $\max\{\frac{1}{2},\frac{k}{n}\}$ approximation ratio
achieved by algorithms in \cite{H12} under many scenarios. For
instance, if $V_M > 0.5V_L$ and $k<0.5n$, our algorithm will perform
better. In a real world, the annual return of an asset will usually
be much less than $100\%$ (i.e. $\frac{V_M}{V_L} \geq 0.5$) while an
asset manager seldom closes his positions on half of its portfolio
holdings (i.e. $\frac{k}{n}<0.5$). Thus, our algorithm will have a
good chance to perform better in practice.

\section{Conclusion} \label{sec:final}
We studied DSHP under the constraint that each asset can take only a
constant number of different values in this paper. We show that if
each asset can take only two values, the problem becomes
polynomial-time solvable. However, even if each asset can take three
different values, DSHP is still NP-hard. A $\frac{V_M}{V_L}$
approximation algorithm is also given under this setting.

There are quite a few interesting open problems raised in this line
of work. First of all, we can extend our model to the multistage
case instead of two stages. It would be interesting to design an
approximation algorithm for that case. Another interesting direction
is to impose some restrictions on the relationships among prices of
all the assets, such that the prices of assets correlate with each
other in some way, and study the complexity in that case.

\end{document}